\documentclass{llncs}
\makeatletter
\renewcommand\subsubsection{\@startsection{subsubsection}{3}{\z@}%
                       {-18\p@ \@plus -4\p@ \@minus -4\p@}%
                       {0.5em \@plus 0.22em \@minus 0.1em}%
                       {\normalfont\normalsize\bfseries\boldmath}}
\makeatother
\usepackage[english]{babel}
\usepackage[utf8]{inputenc}
\usepackage[numbers,sort&compress,merge,round]{natbib}

\usepackage{fullpage}
\usepackage{tabularx}
\usepackage[doublespacing]{setspace}

\usepackage{algorithm}
\usepackage{algorithmicx}
\usepackage{tikz}
\usepackage{wrapfig}
\usepackage{stmaryrd}
\usepackage[colorlinks=true,urlcolor=red,citecolor=blue]{hyperref}

\usepackage{listings}
\usepackage{amsmath}
\usepackage{amsfonts}
\usepackage{amssymb}
\usepackage{textcomp}

\graphicspath{{./figs/biorxiv/}}

\def\DEF{:=}
\def\EQDEF{\overset{\Delta}\Longleftrightarrow}

\def\card#1{|#1|}

\def\B{\mathbb{B}}
\def\M{\mathbb{M}}
\def\R{\mathbb{R}}
\def\nrange#1#2{\{#1, \dots, #2\}}
\def\H{(\B\cup\{*\})}

\def\range#1{[#1]}

\def\dite#1#2{\xrightarrow[\mathrm{#1}]{#2}}
\def\dReach#1#2{\rho^{#2}_{\mathrm{#1}}}

\def\some{\sigma}

\def\us{s}
\def\uf{a1}
\def\ua{a}
\def\up{mp}

\def\ite#1{\dite{\ua}{#1}}

\def\site#1{{\dite{\us}{#1}}}
\def\fite#1{{\dite{\uf}{#1}}}

\def\closure{\nolinebreak\negthickspace{}^*\,}
\def\reach#1{\ite{#1}\closure}

\newcommand{\mpup}{{\scriptstyle \nearrow}}
\newcommand{\mpdw}{{\scriptstyle \searrow}}
\newcommand{\mpDom}{\mathbb P}
\newcommand{\mptob}{\gamma}
\newcommand{\mpite}[1]{\dite{\up}{#1}}
\newcommand{\mpreach}[1]{\dite{\up}{#1}\closure}
\newcommand{\mpReach}[1]{\dReach{\up}{#1}}

\usepackage{color}

\tikzstyle{every matrix}=[ampersand replacement=\&]
\tikzstyle{shorthandoff}=[]
\tikzstyle{shorthandon}=[]

\title{Most Permissive Semantics of Boolean Networks\\(Technical Report)}
\author{Thomas Chatain\inst{1} \and Stefan Haar\inst{1} \and Juraj Kol\v{c}\'ak\inst{1} \and Lo\"ic Paulev\'e\inst{2}}
\institute{
LSV, ENS Paris-Saclay, INRIA, CNRS, France
\and
Université Bordeaux, Bordeaux INP, CNRS, LaBRI, UMR5800\\
F-3304  Talence, France
}

\begin{document}
\sloppy

\maketitle

\begin{abstract}
As shown in \cite{MPBNs}, the usual update modes of Boolean networks (BNs), including
synchronous and (generalized) asynchronous, fail to capture behaviors introduced by multivalued refinements.
Thus, update modes do not allow a correct abstract reasoning on dynamics of biological systems, as
they may lead to reject valid BN models.

This technical report lists the main definitions and properties of the most permissive semantics of
BNs introduced in \cite{MPBNs}.
This semantics meets with a correct abstraction of any multivalued refinements, with any update mode.
It subsumes all the usual updating modes, while enabling new behaviors achievable by more concrete
models.
Moreover, it appears that classical dynamical analyzes of reachability and attractors have a simpler
computational complexity:
\begin{itemize}
    \item reachability can be assessed in a polynomial number of iterations.
The computation of iterations is in NP in the very general case, and is linear when local functions
are monotonic, or with some usual representations of functions of BNs (binary decision diagrams,
Petri nets, automata networks, etc.).
Thus, reachability is in P with locally-monotonic BNs, and P$^{\text{NP}}$ otherwise
        (instead of being PSPACE-complete with update modes);
    \item
deciding wherever a configuration belongs to an attractor is in coNP with locally-monotonic BNs, and
coNP$^{\text{coNP}}$ otherwise (instead of PSPACE-complete with update modes).
\end{itemize}

Furthermore, we demonstrate that the semantics completely captures any behavior achievable with any
multilevel or ODE refinement of the BN;
and the semantics is minimal with respect to this model refinement criteria:
to any most permissive trajectory, there exists a multilevel refinement of the BN which can
reproduce it.

In brief, the most permissive semantics of BNs enables a correct abstract reasoning on dynamics of
BNs, with a greater tractability than previously introduced update modes.
\end{abstract}

\clearpage

\section{Boolean networks}

The Boolean domain is denoted by $\B\DEF\{0,1\}$.
Given a \emph{configuration} \(x\in \B^n\) and \(i \in \range n\), we denote \(x_i\) the
\(i\textsuperscript{th}\) component of \(x\), so that \(x = x_1 \dots x_n\),
and $\bar x$ the complement of $x$, i.e., $\forall i\in\range n$, $\bar x_i=1-x_i$.
Given two configurations $x,y\in \B^n$, the components having a different state are noted
$\Delta(x,y)\DEF\{ i\in \range n\mid x_i\neq y_i\}$.
Symbol \(\wedge\) denotes the logical conjunction, \(\vee\) the disjunction, and \(\neg\) the
negation.
Given a finite set  \(S\), \(\card S\) is its cardinality.

\begin{definition}[Boolean network]
A Boolean network (BN) of dimension $n$ is a function
\(f:\B^n\to\B^n\).
For each \(i\in\range n\), \(f_i:\B^n\to\B\) denotes the local function of its $i$th component.
\end{definition}

\begin{definition}[Locally-monotonic BN]
    A BN \(f:\B^n\to\B^n\) is \emph{locally monotonic} whenever for each component
    \(i\in\nrange 1n\), there exists an ordering of components
    \(\preceq^i\in \{\leq,\geq\}^n\)
    such that
    \(
        \forall x,y\in\B^n,
        (x_1\preceq^i_1 y_1 \wedge ... \wedge x_n\preceq^i_n y_n) \Rightarrow
        f_i(x) \leq f_i(y)
    \).
\end{definition}

\begin{example}
    The BN \(f\) of dimension \(3\) defined as
\begin{align*}
f_1(x) &= x_3 \wedge (\neg x_1 \vee \neg x_2)\\
f_2(x) &= x_3 \wedge x_1\\
f_3(x) &= x_1\vee x_2\vee x_3\enspace,
\end{align*}
is locally monotonic, for instance with
\(\preceq^1=(\geq,\geq,\leq)\) and \(\preceq^2=\preceq^3=(\leq,\leq,\leq)\).
\end{example}

\section{Most Permissive Boolean Networks}

\subsection{Definitions}

We give two different definitions which are equivalent in term of reachability properties.
The first one introduces dynamic states, the second one relies on the computation of
hypercubes.

\subsubsection{With dynamic states}

A most-permissive configuration assigns to each BN component one state among four, noted
\(\mpDom\DEF\{0, \mpup, \mpdw, 1\}\).
The possible binary
interpretations of a configuration \(x\in\mpDom^n\) are denoted by
\begin{equation}
    \mptob(x)\DEF\{\tilde x\in\B^n\mid \forall i\in\range n, x_i\in\B\Rightarrow \tilde x_i=x_i\}
    \enspace.
\end{equation}

The semantics is defined as an irreflexive binary relation between configurations in \(\mpDom^n\):
\begin{definition}[Most permissive semantics \(\mpite f\)]\label{def:mp}
    \begin{equation*}
        \begin{aligned}
    \forall x,y\in\mpDom^n,\quad
    x \mpite f y&\EQDEF \exists i\in \range n: \Delta(x,y)=\{i\} \\
                & \wedge y_i=
  \begin{cases}
      \mpup & \text{if }x_i\neq 1 \wedge \exists \tilde x\in\mptob(x): f_i(\tilde x)\\
      1 & \text{if }x_i=\mpup\\
      \mpdw & \text{if }x_i\neq 0 \wedge \exists \tilde x\in\mptob(x): \neg f_i(\tilde x)\\
      0 &\text{if }x_i=\mpdw
  \end{cases}
  \end{aligned}
  \end{equation*}
\end{definition}

The set of binary configurations reachable from \(x\in\B^n\) with the most
permissive semantics is given by
\begin{equation}
    \mpReach f (x) \DEF \{ y\in\B^n\mid x\mpite f\closure y\}
    \enspace.
\end{equation}

The following figure shows the automaton of the state change of a component $i$ in the most permissive semantics,
following notations of Def.~\ref{def:mp}.
The labels $f_i(\tilde x)$ and $\neg f_i(\tilde x)$ on edges are the conditions for firing the
transitions, where \(\tilde x\in\mptob(x)\); the label
$\epsilon$ indicates transitions that can be done without condition:
\begin{center}
\begin{tikzpicture}[state/.style={circle,draw},node distance=2cm,alias/.style={font=\tt}]
\node[state] (on) {1};
\node[state,below left of=on] (onoff) {\(\mpup\)};
\node[state,below right of=onoff] (off) {0};
\node[state,above right of=off] (offon) {\(\mpdw\)};
\path[->,>=latex]
    (off) edge node[below,sloped] {$f_i(\tilde x)$} (offon)
    (on) edge node[above,sloped] {$\neg f_i(\tilde x)$} (onoff)
    (offon) edge node[above,sloped] {$\epsilon$} (on)
    (onoff) edge node[below,sloped] {$\epsilon$} (off)
    (offon) edge[bend right=10] node[above, sloped] {$\neg f_i(\tilde x)$} (onoff)
    (onoff) edge[bend right=10] node[below, sloped] {$f_i(\tilde x)$} (offon)
;
\end{tikzpicture}
\end{center}
With the given definition, only one automaton is updated at a time.
However, it is equivalent to allow any number of simultaneous changes, as long as fully asynchronous
updates are considered.

Given a configuration \(x\in\mpDom^n\), one can remark that as long as only transitions towards
dynamic states \(\mpup\) or \(\mpdw\) are performed, then the set of binary interpretations
\(\mptob\) is growing.
As a consequence, the ordering of such transitions does not matter.
\begin{proposition}
    \label{pro:mp-monotonicity}
Given a BN $f$ of dimension $n$, $\forall x,y\in\mpDom^n$ such that $x\mpite f y$ and
$\forall j\in\Delta(x,y): y_j\notin\B$, $\mptob(x)\subseteq \mptob(y)$.
\end{proposition}

Given a configuration \(x\in\B^n\), if we consider any reachable configuration where changed
components are in an dynamic state, and from which there is no more transitions from binary states
towards dynamic states, then the set of binary interpretation of this later configuration
includes the set of all binary configurations reachable from \(x\):
\begin{proposition}
Given a BN \(f\) of dimension \(n\) and a binary configuration \(x\in\B^n\),
let us consider a configuration \(z\in\mpDom^n\) such that
\(x\mpreach f z\),
$\forall i\in\Delta(x,z): z_i\notin\B$, and there is no \(z'\in\mpDom^n\)
such that \(z\mpite f z\) with for \(j\in\Delta(z,z')\),
\(z_j\in\B\) and \(z'_j\notin\B\),
then
\(\mpReach f(x)\subseteq\mptob(z)\).
\end{proposition}

\subsubsection{With hypercubes}

The dynamic states might suggest that the most permissive semantics is close to multivalued
networks with \(4\) states.
However, notice that states \(\mpDom\) are not totally ordered by the transitions, as it is required
by multivalued networks.

We give here an equivalent definition of \(\mpReach f\) which does not relies on these dynamic
states, but on the computation of hypercubes closed by \(f\).
An hypercube within \(\B^n\) has a set of components being fixed to a Boolean state, and the others
being free (noted with \(*\)).
\begin{definition}[Hypercube]
    An \emph{hypercube} \(h\) of dimension \(n\) is a vector in $\H^n$.
    The set of its associated configurations is denoted by
    \(c(h)\DEF \{x\in\B^n\mid \forall i\in\range n, h_i\neq *\Rightarrow x_i=h_i\}\).

    Given two hypercubes \(h,h'\in\H^n\), \(h\) is \emph{smaller} than \(h'\)
    if and only if
    \(\forall i\in\range n, h'_i\neq*\Rightarrow h_i=h'_i\).
    An hypercube is \emph{minimal} if there is no different hypercubes smaller than it.

    An hypercube \(h\) is \emph{closed} by \(f\) whenever for each configuration \(x\in c(h)\),
    \(f(x)\in c(h)\).
\end{definition}

An hypercube closed by \(f\) is also known as a \emph{trap space}; if it is minimal, it is a
\emph{minimal trap space}.

We generalize the notion of closure by allowing restricting the set of components which should be
closed.
\begin{definition}[\(K\)-closed hypercube]
    Given a subset of components \(K\subseteq\range n\),
    an hypercube \(h\in\H^n\) is \(K\)-closed by \(f\) whenever for each configuration \(x\in
    c(h)\),
    for each component \(i\in K\),
    \(h_i\in\{*,f_i(x)\}\).
\end{definition}

Remark: an hypercube is closed if and only if it is \(\range n\)-closed.

\begin{example}
    Let us consider the BN \(f:\B^3\to\B^3\) with
    \(f_1(x) \DEF \neg x_2\), $f_2(x) \DEF\neg x_1$, et $f_3(x) \DEF \neg x_1 \wedge x_2$.
The hypercube \(01*\) is closed by \(f\), with \(c(01*)=\{010,011\}\).
The hypercube \(1\!*\!0\) is the smallest hypercube \(\{2,3\}\)-closed by \(f\) containing
\(110\); it is not closed by \(f\), nor the smallest hypercube \(\{2,3\}\)-closed by \(f\)
containing \(100\).
\end{example}

Starting from a binary configuration \(x\in\B^n\), the most permissive semantics can be expressed
using the computation of smallest hypercubes containing \(x\) and which are \(K\)-closed by \(f\),
for every \(K\):
\begin{itemize}
    \item \(x\) is the unique hypercube  $\emptyset$-closed by $f$ containing $x$;
    \item the change of state of component $i\in\range n$ to $\mpup$ or $\mpdw$
produces a configuration $x'$ where $\mptob(x')$ correspond to the hypercube $h\in\H^n$ with
$h_i=*$ and for each other component $j\in\range n, j\neq i$,
$h_j=x_j$.
Thus, $h$ is the smallest hypercube $\{i\}$-closed by $f$ and containing $x$;
\item
    by considering only the change of states towards $\mpup$ and $\mpdw$,
    the most permissive semantics progressively enlarges the hypercubes along the modified
    components, and each step results in a smallest hypercube $K$-closed by \(f\) and containing
    \(x\), for every $K\subseteq\range n$.
\end{itemize}

With the most permissive semantics, the change of state of a component from
a dynamic to a Boolean state is without condition, and is solely determined by its
current dynamic state: $1$ from $\mpup$ and $0$ from $\mpdw$.
Thus, starting from an initial configuration which is binary,
a component can be in the state $\mpup$ only if a preceding configuration
$x'\in\mpDom^n$ was such that $\exists z\in\mptob(x')$ with $f_i(z)=1$ (resp. $\mpdw$ if $f_i(z)=0$).

The following proposition establishes the correspondence with the initial definition with
dynamic states:
\begin{proposition}\label{prop:hypercube_reachability}
    Given a BN \(f\) of dimension \(n\) and two configurations
    \(x,y\in\B^n\),
    \(y\in\mpReach f(y)\) if and only if
    there exists
    \(K\subseteq \range n\)
    such that
    the smallest \(K\)-closed hypercube \(h\) and containing \(x\)
    verifies
    (1) \(y\in c(h)\), and
    (2) \(\forall i\in K\), there exists a configuration
    \(z\in c(h)\) such that \(f_i(z)=y_i\).
\end{proposition}

\begin{example}
\def\myscale{1}
Here below are examples of smallest \(K\)-closed hypercubes containing la configuration \(000\)
(left), \(010\) (top right), and \(011\) (bottom right) for the BN \(f\) of dimension \(3\)
    defined by
    \(f_1(x) \DEF \neg x_2\),
    \(f_2(x) \DEF\neg x_1\),
    \(f_3(x) \DEF \neg x_1 \wedge x_2\).
    Configurations belonging to the hypercube are highlighted in bold;
    these verifying the reachability property are boxed.
    The hypercube \(011\)  is only one which is closed by \(f\) and minimal.

\centering
\begin{tabular}{cc|c}
    \scalebox{\myscale}{
    \begin{tikzpicture}
\matrix[column sep=0.8cm, row sep=1cm,gray] {
\node (s010) {$010$}; \&
\node (s110) {$110$};
\\
\node[black,font=\bf,draw] (s000) {$\mathbf{000}$}; \&
\node (s100) {$100$};
\\
};
\matrix[column sep=0.8cm, row sep=1cm,shift={(1cm,0.6cm)},gray] {
\node (s011) {$011$}; \&
\node (s111) {$111$};
\\
\node (s001) {$001$}; \&
\node (s101) {$101$};
\\
};
\path[gray]
(s000) edge (s010) edge (s100) edge[densely dashed] (s001)
(s110) edge (s010) edge (s100) edge (s111)
(s001) edge[densely dashed] (s011) edge[densely dashed] (s101)
(s111) edge (s011) edge (s101)
(s010) edge (s011)
(s100) edge (s101)
;
    \end{tikzpicture}}&
    \scalebox{\myscale}{\begin{tikzpicture}
\matrix[column sep=0.8cm, row sep=1cm,gray] {
\node (s010) {$010$}; \&
\node (s110) {$110$};
\\
\node[black,font=\bf] (s000) {$\mathbf{000}$}; \&
\node[black,font=\bf,draw] (s100) {$\mathbf{100}$};
\\
};
\matrix[column sep=0.8cm, row sep=1cm,shift={(1cm,0.6cm)},gray] {
\node (s011) {$011$}; \&
\node (s111) {$111$};
\\
\node (s001) {$001$}; \&
\node (s101) {$101$};
\\
};
\path[gray]
(s000) edge (s010) edge (s100) edge[densely dashed] (s001)
(s110) edge (s010) edge (s100) edge (s111)
(s001) edge[densely dashed] (s011) edge[densely dashed] (s101)
(s111) edge (s011) edge (s101)
(s010) edge (s011)
(s100) edge (s101)
;
    \end{tikzpicture}}&
    \scalebox{\myscale}{
    \begin{tikzpicture}
\matrix[column sep=0.8cm, row sep=1cm,gray] {
\node[black] (s010) {$\bf 010$}; \&
\node (s110) {$110$};
\\
\node (s000) {$000$}; \&
\node (s100) {$100$};
\\
};
\matrix[column sep=0.8cm, row sep=1cm,shift={(1cm,0.6cm)},gray] {
\node[black,draw] (s011) {$\bf011$}; \&
\node (s111) {$111$};
\\
\node (s001) {$001$}; \&
\node (s101) {$101$};
\\
};
\path[gray]
(s000) edge (s010) edge (s100) edge[densely dashed] (s001)
(s110) edge (s010) edge (s100) edge (s111)
(s001) edge[densely dashed] (s011) edge[densely dashed] (s101)
(s111) edge (s011) edge (s101)
(s010) edge (s011)
(s100) edge (s101)
;
    \end{tikzpicture}}
    \\
    \(K=\emptyset\)
    &
    \(K=\{1\}\)
    &
    \(K=\{1,2,3\}\)
    \\

    \(000\)
    &
    \(*00\)
    &
    \(01*\)

    \\\cline{3-3}
    &&
    \\
    \scalebox{\myscale}{\begin{tikzpicture}
\matrix[column sep=0.8cm, row sep=1cm,gray] {
\node[black,font=\bf,draw] (s010) {$\mathbf{010}$}; \&
\node[black,font=\bf,draw] (s110) {$\bf 110$};
\\
\node[black,font=\bf,draw] (s000) {$\mathbf{000}$}; \&
\node[black,font=\bf,draw] (s100) {$\mathbf{100}$};
\\
};
\matrix[column sep=0.8cm, row sep=1cm,shift={(1cm,0.6cm)},gray] {
\node (s011) {$011$}; \&
\node (s111) {$111$};
\\
\node (s001) {$001$}; \&
\node (s101) {$101$};
\\
};
\path[gray]
(s000) edge (s010) edge (s100) edge[densely dashed] (s001)
(s110) edge (s010) edge (s100) edge (s111)
(s001) edge[densely dashed] (s011) edge[densely dashed] (s101)
(s111) edge (s011) edge (s101)
(s010) edge (s011)
(s100) edge (s101)
;
    \end{tikzpicture}}&
    \scalebox{\myscale}{\begin{tikzpicture}
\matrix[column sep=0.8cm, row sep=1cm,gray] {
\node[black,font=\bf,draw] (s010) {$\mathbf{010}$}; \&
\node[black,font=\bf,draw] (s110) {$\bf 110$};
\\
\node[black,font=\bf,draw] (s000) {$\mathbf{000}$}; \&
\node[black,font=\bf,draw] (s100) {$\mathbf{100}$};
\\
};
\matrix[column sep=0.8cm, row sep=1cm,shift={(1cm,0.6cm)},gray] {
\node[black,draw] (s011) {$\bf 011$}; \&
\node[black,draw] (s111) {$\bf 111$};
\\
\node[black,draw] (s001) {$\bf 001$}; \&
\node[black,draw] (s101) {$\bf 101$};
\\
};
\path[gray]
(s000) edge (s010) edge (s100) edge[densely dashed] (s001)
(s110) edge (s010) edge (s100) edge (s111)
(s001) edge[densely dashed] (s011) edge[densely dashed] (s101)
(s111) edge (s011) edge (s101)
(s010) edge (s011)
(s100) edge (s101)
;
    \end{tikzpicture}}

    &\scalebox{\myscale}{
    \begin{tikzpicture}
\matrix[column sep=0.8cm, row sep=1cm,gray] {
\node (s010) {$010$}; \&
\node (s110) {$110$};
\\
\node (s000) {$000$}; \&
\node (s100) {$100$};
\\
};
\matrix[column sep=0.8cm, row sep=1cm,shift={(1cm,0.6cm)},gray] {
\node[black,draw] (s011) {$\bf011$}; \&
\node (s111) {$111$};
\\
\node (s001) {$001$}; \&
\node (s101) {$101$};
\\
};
\path[gray]
(s000) edge (s010) edge (s100) edge[densely dashed] (s001)
(s110) edge (s010) edge (s100) edge (s111)
(s001) edge[densely dashed] (s011) edge[densely dashed] (s101)
(s111) edge (s011) edge (s101)
(s010) edge (s011)
(s100) edge (s101)
;
    \end{tikzpicture}}

    \\
    \(K=\{1,2\}\)
    &
    \(K=\{1,2,3\}\)
    &
    \(K=\{1,2,3\}\)
    \\
    \(**0\)
    &
    \(***\)
    &
    \(011\)
    \\
\end{tabular}
\end{example}

\subsection{Relation with quantitative refinements}

Multivalued networks (MNs) are a generalization of BNs where the components
can take values in a finite discrete domain. Let us denote the possible
values as \(\M \DEF \{0, 1, \dots, m\}\) for some integer \(m\).
Without loss of generality, we assume the same domain of values for all the components.

\begin{definition}[Multivalued network]
A \emph{multivalued network} (MN) of dimension $n$ over a value range \(\M =
\{0, 1, \dots, m\}\) is a function
\(F: \M^n\to\{-1,0,1\}^n\).
\end{definition}
A configuration of a MN of dimension \(n\) is a vector \(x \in \M^n\).
Given two configurations $x,y\in \M^n$, the components that differ are noted
$\Delta(x,y)\DEF\{ i\in \range n\mid x_i\neq y_i\}$.

\begin{definition}[Asynchronous semantics]
Given a multivalued network $F$, the binary irreflexive relation
$\ite F\,\subseteq \M^n\times\M^n$
is defined as:
\begin{equation*}
        x\ite F y \EQDEF \forall i\in\Delta(x, y), y_i = x_i + F_i(x)
    \enspace.
\end{equation*}
We write $\reach F$ for the transitive closure of $\ite F$.

\end{definition}

We now define a notion of \emph{multivalued refinement} of a BN, which formalizes the intuition that
the value changes defined by the multivalued network are compatible with those of the BN.
The refinement criteria relies on a \emph{binarization} of the multivalued configuration.
An appropriate binarization necessarily quantifies \(0\) as Boolean \(0\) and \(m\) as \(1\), and is
free for the other dynamic states.
Let us denote by $\beta(x)$ the set of possible binarization of configuration $x\in\M^n$:
\begin{equation}\label{eq:beta}
\beta(x) \DEF \{ x'\in\B^n\mid \forall i\in\range n, x_i=0\Rightarrow x'_i=0
    \wedge x_i=m\Rightarrow x'_i=1\}\enspace.
\end{equation}

\begin{definition}[Multivalued refinement]
  \label{def:multivalued-refinement}
  A multivalued network \(F\) of dimension $n$ over a value range \(\M\)
     \emph{refines} a BN
  \(f\) of equal dimension \(n\) if and only if for every configuration \(x \in \M^n\)
  and every \(i \in \range n\),
  \[
  F_i(x) > 0 \Rightarrow \exists x' \in \beta(x): f_i(x') = 1 \wedge 
  F_i(x) < 0 \Rightarrow \exists x' \in \beta(x): f_i(x') = 0
  \enspace.
  \]
\end{definition}

This characterization of refinement can be readily extended to ODEs:
similarly to multivalued networks, ODEs specify the derivative of the (positive) real value of each
component along the continuous time \(t\):
\begin{equation}
    \frac{d\mathbb F(t,x)}{dt} = \mathcal F(x)
    \qquad
    \text{with }
     \mathcal{F} : \R_{\geq 0}^n\to \R^n
     \enspace.
\end{equation}
Here, \(\mathcal F(x)\) is the derivative of \(\mathbb F(t,x)\) along time \(t\) in function of
continuous configurations \(x\); \(\mathbb F\) being usually unknown.
ODEs can be seen thus be seen as MNs with \(m\) going to infinity and with synchronous semantics:
\(\mathcal F\) model the simultaneous evolution of all the components.

The admissible binarizations \(\beta\) should be slightly adapted to reflect the absence \textit{a
priori} of maximum value: 
\(\beta(x) \DEF \{ x'\in\B^n\mid \forall i\in\range n, x_i=0\Rightarrow x'_i=0\}\).
Then, the definition of refinement is identical.

\subsubsection{Completeness}
Let us consider a BN \(f\) of dimension \(n\) and any multivalued refinement \(F\) with \(m\)
values.
A \emph{most-permissive interpretation} of a multivalued configuration is a configuration in
\(\mpDom^n\) where components having extreme states in the multivalued configuration have the
corresponding extreme states in the most permissive configuration, and otherwise are either
\(\mpup\) or \(\mpdw\).
Let us denote these interpretations by
\begin{equation}
\alpha(x)\DEF\{ \hat x\in\mpDom^n\mid x_i=0\Leftrightarrow \hat x_i=0 \wedge x_i=m\Leftrightarrow
\hat x_i=m\}
\end{equation}

Then, Theorem~\ref{thm:mp-correctness} states that for any asynchronous transition from \(x\) to \(y\)
(\(x\ite F y\)), there is a most permissive trajectory from any corresponding most permissive configuration \(\hat x\in\alpha(x)\)
to a configuration \(\hat y\in\alpha(y)\) where the state of
each component is consistent with the changes between \(x\) and \(y\).
\begin{theorem}\label{thm:mp-correctness}
    Given a BN \(f\) of dimension \(n\),
    for any multivalued network \(F:\M^n\to \{-1,0,1\}^n\) being a refinement of \(f\),
    \[\forall x,y\in\M^n, \quad
        x\ite F y \Longrightarrow
        \forall \hat x\in\alpha(x),\exists \hat y\in\alpha(x):
        \hat x\mpreach f \hat y
        \text{ with } \forall i\in\range n,
    \hat y_i =  \begin{cases}
        \mpup & \text{if }y_i > x_i \wedge y_i < m\\
        \mpdw & \text{if }y_i < x_i \wedge y_i > 0\\
        0 & \text{if }y_i=0 \\
        1 & \text{if }y_i=m \\
        \hat x_i &\text{otherwise.}
    \end{cases}
\]
\end{theorem}

\begin{proof}
    From MN semantics, for each component \(i\in\Delta(x,y)\), whenever \(y_i > x_i\) (resp. \(y_i <
x_i\)), necessarily \(F_i(x) > 0\) (resp. \(F_i(x) < 0\)).
From the refinement property, there exists a binarization \(x'\in\beta(x)\) such that \(f_i(x)=1\) (resp.
\(f_i(x)=0\)).
Now remark that for any \(\hat x\in\alpha(x)\), \(x'\in\beta(\hat x)\).
Therefore, for each component \(i\in\Delta(x,y)\), if \(y_i>x_i\) and \(\hat x_i\neq \mpup\), the
\(i\) can change to state \(\mpup\), and if \(y_i<x_i\) and \(\hat x_i\neq \mpdw\), the
\(i\) can change to state \(\mpdw\).
By Proposition~\ref{pro:mp-monotonicity}, these transitions can be applied in any order; let us
denote by \(z\) the obtained configuration.
Finally, for each component \(i\in\Delta(x,y)\) where \(y_i=0\) (resp. \(y_i=m\)),
remark that \(z_i=\mpdw\) (resp. \(z_i=\mpup\)), thus it can change to state \(0\) (resp. \(1\)),
in any order.
Therefore, \(\hat x\mpreach f \hat y\).
\end{proof}

Remark that the theorem considers asynchronous transition, which includes any restrictions
(synchronous, fully asynchronous, sequential, ...).

As the proof relies solely on the sign of the derivative of the refinement of \(f\), the property
extends to ODE refinements, which can be seen as MN with \(m\) to infinity and with synchronous
semantics.
The function \(\alpha\) then becomes
\begin{equation}
 \alpha(x)\DEF\{
 \hat x \in \mpDom^n \mid
 \forall i\in\range n,
 x_i=0 \Leftrightarrow \hat x_i=0\wedge \hat x_i\neq 1\}
 \end{equation}
\begin{corollary}\label{cor:mp-correctness}
    For any ODE system \(\mathcal F: \mathbb R_{\geq 0}^n\to\mathbb R^n\) refining a BN \(f\) of
    dimension \(n\),
    \[
    \forall x\in\mathbb R_{\geq 0}^n,
    \forall \hat x\in\alpha(x),
    \quad
    \hat x\mpreach f\hat y
    \quad
    \text{with }\forall i\in\range n,
    \hat y_i =  \begin{cases}
        \mpup & \text{if }F_i(x) > 0\\
        \mpdw & \text{if }F_i(x) < 0\wedge x_i>0\\
        \hat x_i &\text{otherwise.}
    \end{cases}
\]
\end{corollary}

Remark that a BN $f$ is a multivalued refinement of itself with \(\M=\B\) and for each \(i\in\range
n\), \(F_i(x) = 1 \text{ if }f_i(x), -1 \text{ otherwise}\).
Therefore another corollary of the above theorem is that the most permissive semantics of BNs
simulates the asynchronous semantics of \(f\):
\begin{corollary}
    Given a BN \(f\) of dimension \(n\),
    \[
        \forall x,y\in\B^n,\quad
        x \reach f y \Longrightarrow x \mpreach f y
        \enspace.
    \]
\end{corollary}

Thus, the number of attractors with the most permissive semantics is at most the number of attractor with update semantics.

\subsubsection{Minimality}
Whereas complete, one should wonder whether the most permissive semantics introduce spurious
behaviors.
We prove in this section that the most permissive semantics is the tightest abstraction of
multivalued refinements with respect to reachability properties.

First, Proposition~\ref{pro:mp-extreme-minimality} ensures that if there exists a most-permissive
trajectory between two Boolean configurations \(x\) and \(y\in\B^n\), then there exists a multilevel refinement of the BN
which allows an asynchronous trajectory between corresponding configurations \(m.x\) and \(m.y\)
with \(m=2\).
The idea is to construct a MN which can reproduce the shortcut trajectory, with dynamic states
identified to an intermediate state \(1\) of the MN:
in a first phase, components increase to \(1\) (possibly fully-asynchronously), then a last
synchronous step leads to the target \(2.y\) configuration.

Then, we introduce the notion of trace refinement witch matches most permissive trajectories with MN
asynchronous trajectories having coherent successions of states, both with respect to admissible
most-permissive interpretation, and with respect to derivatives: whenever a component \(i\) changes to the dynamic
state \(\mpup\) (resp. \(\mpdw\)), \(F_i\) is positive (resp. negative) in the corresponding
multivalued configuration.
Theorem~\ref{thm:mp-trace-minimality} establishes for any most permissive trajectory, there
exists a MN refinement with \(m=3\) which admits a matching asynchronous trajectory.

Therefore, the most permissive semantics introduces no spurious behavior with respect to
the admissible refinements of a BN \(f\).

\begin{proposition}\label{pro:mp-extreme-minimality}
For any BN \(f\) of dimension \(n\) and any pair of configurations \(x, y\in\B^n\),
if \(y\) is reachable from \(x\) with the most permissive semantics,
then there exists a MN \(F\) with \(m\) values which is a refinement of \(f\) and
where \(m.y\) is reachable from \(m.x\) with the asynchronous semantics.
\end{proposition}
\begin{proof}
Let \(K\subseteq\range n\) be the smallest subset of components verifying Proposition \ref{prop:hypercube_reachability}.
We now define a sequence of configurations \(x, x',\dots,x^{(|K|)}\in\mpDom^n\) to be arbitrary such that
\(\forall 0 < i \leq |K|\), \(x^{(i-1)} \mpite f x^{(i)}\) and
\(j \in \Delta(x^{(i-1)}, x^{(i)}) \Longrightarrow j\in K \wedge x^{(i-1)}_j \in \B \wedge x^{(i)}_j \notin \B\).
Note that such a sequence is guaranteed to exist thanks to \(K\) being minimal.

We define another sequence of configurations \(z, z', \dots,z^{(|K|)}\in\{0,1,2\}^n\) as the multivalued equivalent of \(x, x',\dots,x^{(|K|)}\):
\(\forall 0 \leq i \leq |K|\) and \(\forall j\in\range n\), \(x^{(i)}_j \in \B \Longrightarrow z^{(i)}_j = 2.x^{(i)}_j\) and \(x^{(i)}_j \notin \B \Longrightarrow z^{(i)}_j = 1\).

We now construct the coveted MN \(F\) with \(3\) values. based on \(z, z',\dots,z^{(|K|)}\) as follows:
\begin{itemize}
	\item For any \(0\leq i<|K|\), \(F(z^{(i)}) = z^{(i+1)}-z^{(i)}\).
	\item \(F(z^{|K|}) = 2.y - z^{|K|}\). (\(2.y - z^{|K|} \in \{-1,0,1\}^n\) thanks to \(y\) being in the smallest \(K\)-closed hypercube containing \(x\).)
	\item For any other \(z \in \{0,1,2\}^n\), \(F(z) = 0^n\).
\end{itemize}

Clearly, \(2.y\) is reachable from \(2.x = z\) in \(F\) with synchronous semantics.
What remains to be proven is that \(F\) is a refinement of \(f\).
Nothing needs to be shown for cases when \(F\) returns \(0\),
let thus first \(0 \leq i <|K|\) and \(\{j\} = \Delta(z^{(i)},z^{(i+1)})\).
Let us further assume \(F_j(z^{(i)}) = 1\) as the \(F_j(z^{(i)}) = -1\) case is symmetric.
We need to show \(\exists \tilde{x}\in\beta(z^{(i)})\) such that \(f(\tilde{x})\).
By definition, \(x^{(i +1)} = \mpup\), thus \(\exists \tilde{x}\in\gamma(x^{(i)})\) such that \(f(\tilde{x})\).
Since for any \(j\in\range n\), \(z^{(i)}_j = 1\) exactly when \(x^{(i)}_j \notin\B\), we have \(\gamma(x^{(i)}) \subseteq \beta(z^{(i)})\).

Finally, let us consider \(z^{|K|}\).
We need to show \(\forall j\in \Delta(z^{(|K|)}, 2.y)\), \(\exists \tilde{x}\in\beta(z^{(|K|)})\), \(f(\tilde{x}) = y_j\).
By definition, we have \(\Delta(z^{(|K|)}, 2.y) = K\).
Since \(K\) verifies Property \ref{prop:hypercube_reachability}, we know \(\forall j\in K\),  \(\exists \tilde{x}\in c(h)\), \(f(\tilde{x}) = y_j\),
where \(h\) is the smallest \(K\)-closed hypercube containing \(x\).
By definition of \(x^{(|K|)}\), \(\forall j\in K\), \(x^{(|K|)}_j \notin \B\) and thus \(c(h)\subseteq \gamma(x^{(|K|)})\).
Furthermore, since for any \(j\in\range n\), \(z^{(|K|)}_j = 1\) exactly when \(x^{(|K|)}_j \notin\B\), we have \(\gamma(x^{(|K|)}) \subseteq \beta(z^{(|K|)})\).

\end{proof}

\begin{definition}[Trace Refinement]
  \label{def:trace_refinement}
  Given a BN $f$ of dimension $n$ and a multivalued refinement $F:\M^n\rightarrow\{-1,0,1\}^n$ of $f$.
Let $x,x',\dots, x^{(k)} \in \mpDom^n$ be a finite sequence of configurations such that
\(\forall 0 < i \leq k\),
\(x^{(i-1)} \mpite f x^{(i)}\) (finite trace of \(f\) with the most permissive semantics).

  Then a finite sequence $y,y',\dots,y^{(l)}\in \M^n$ such that \(\forall 0 < i \leq l\),
\(y^{(i-1)} \ite F y^{(i)}\),
  is a trace refinement of $x,x',\dots, x^{(k)}$ if there exists a function
  $\kappa: \{0,\dots,k\} \rightarrow \{0,\dots,l\}$ (trace refinement function)
  satisfying the following requirements:
  \begin{enumerate}
    \item $\kappa$ is non-decreasing, i.e. $i < j \Longrightarrow \kappa(i) \leq \kappa(j)$;
    \item $\kappa(0) = 0$ and $\kappa(k) = l$;
    \item $\forall j \in\range n$, $x_j = y_j$
      and for each $0 < i\leq k$,
      $(x^{(i)}_j = 0 \Longrightarrow y^{(\kappa(i))}_j < {m}) \wedge
        (x^{(i)}_j = 1 \Longrightarrow  y^{(\kappa(i))}_j > 0)$;
    \item For each $0 < i \leq k$ such that $x^{(i)}_j \notin \B$
      where $\{j\} = \Delta(x^{(i-1)}, x^{(i)})$,
      $x^{(i)}_j = \mpup \Longrightarrow {F}_{j}(y^{(\kappa(i - 1))})= 1$ and $x^{(i)}_j = \mpdw \Longrightarrow {F}_{j}(y^{(\kappa(i - 1))})= -1$.
  \end{enumerate}
\end{definition}

\begin{theorem}\label{thm:mp-trace-minimality}
For any BN \(f\) of dimension \(n\) and any sequence of configurations \(x,x',\dots, x^{(k)} \in \mpDom^n\) such that
\(x\in\B^n\) and
\(\forall 0 < i \leq k\),
\(x^{(i-1)} \mpite f x^{(i)}\),
there exists a MN \(F:\M^n\rightarrow\{-1,0,1\}^n\) which is a refinement of \(f\)
and has a trace refinement \(y,y',\dots,y^{(l)}\in \M^n\) of \(x,x',\dots, x^{(k)}\).
\end{theorem}
\begin{proof}
  We construct $F$ and $y,y',\dots,y^{(l)}$ iteratively along the sequence $x,x',\dots, x^{(k)}$.
  For each step $i \in \{0,\dots,k\}$ we maintain
  that the constructed network $F$ is a refinement of of $f$
  and $y,y',\dots,y^{(l_i)}$ is a trace refinement of $x,x',\dots, x^{(i)}$.

  Let us first construct our initial $F$ and $y$ (for $i = 0$).
  We define the MN $F$ with ${m} = 3$
  as follows:
  \begin{align*}
    \forall z \in {\B}^{n},
    \forall j \in \range n,
    \begin{cases}
        {f}_{j}(z)=0 &\Longrightarrow
        \forall z' \in {\Pi}_{i = 1}^{n} \{2 \cdot {z}_{i}, (2 \cdot {z}_{i} + 1)\},
          {F}_{j}(z') = -1\\
          {f}_{j}(z)=1 &\Longrightarrow
        \forall z' \in {\Pi}_{i = 1}^{n} \{2 \cdot {z}_{i}, (2 \cdot {z}_{i} + 1)\},
          {F}_{j}(z') = 1
      \end{cases}
  \end{align*}

  We first show that $F$ is indeed a refinement of $f$.
  Let $z \in \M^{n}$ and $j \in \range n$ be arbitrary such that
  $F(z) = -1$ as the case of $F(z) = 1$ is symmetric.
  We want to show $\exists z' \in \beta(z)$ such that $f(z')=0$.

  Consider the state $z'$ defined as follows:
  \[
    \forall j \in \range n,
      z'_{j} = \frac{{z}_{j} - ({z}_{j} \bmod 3)}{3}
  \]
  Surely such state is a binarization of $z$, $z' \in \beta(z)$.
  Moreover, $f(z')=0$ as by definition of $F$, $f(z') \Longrightarrow
    F(z) = 1$ leads to a contradiction.

  Let us define \(y = 3.x\):
  it is trivially a trace refinement of $x$ with the trace refinement function
  $\kappa: 0 \mapsto 0$.

  We now iterate over $i \in \{1, \dots, k\}$,
  adjusting $F$, $y,y',\dots,y^{(l_i)}$ and $\kappa$ to ensure $y,y',\dots,y^{(l_i)}$ is a trace refinement of
  $x,x',\dots, x^{(i)}$.
  Moreover, we maintain that no transition increases any component value beyond $2$
  or decreases below $1$ along $y',\dots,y^{(l_i)}$ and ensure that $\forall j \in \range n$,
  $x^{(i)}_j = \mpdw \Longrightarrow y^{(\kappa(i))}_j= 1$
  and $x^{(i)}_j = \mpup \Longrightarrow y^{(\kappa(i))}_j = 2$.

  Let $\{e\} = \Delta(x^{(i-1)}, x^{(i)})$
  and let $l_{i-1}$ denote the current length of the sequence of configurations \(y,y',\dots,y^{(l_{i-1})}\).
  We modify $F$ and extend  \(y,y',\dots,y^{(l_{i-1})}\) based on the value of $x^{(i)}_e$:
  \begin{itemize}
    \item $x^{(i)}_e \notin \B$.
      Let us assume $x^{(i)}_e = \mpdw$ without loss of generality,
      as the construction is symmetric for $x^{(i)}_e = \mpup$.

      First, we extend  \(y,y',\dots,y^{(l_{i-1})}\) based on \(y^{(l_{i-1})}_e\):
      \begin{itemize}
        \item $y^{(l_{i-1})}_e = 3$, $y^{(l_{i-1}+1)} := z \wedge y^{(l_{i-1}+2)} := z'$;
        \item $y^{(l_{i-1})}_e = 2$, $y^{(l_{i-1}+1)} := z'$;
      \end{itemize}
      where $z$ and $z'$ are equal to $y^{(l_{i-1})}$ but ${z}_{e} = 2$ and ${z}'_{e} = 1$.
      The trace refinement function is adjusted accordingly,
	$y^{(l_{i-1})}_e = 3 \Longrightarrow \kappa: i \mapsto l_{i-1}+2 = l_i$ and $y^{(l_{i-1})}_e = 2 \Longrightarrow \kappa: i \mapsto l_{i-1}+1 = l_i$.

      If \(\forall 0 < j \leq l_i\), \(y^{(j-1)} \ite F y^{(j)}\), we are done.
      Otherwise, we modify ${F}_{e}(y^{(l_{i-1})}) := -1$ and,
      if necessary, also ${F}_{e}(z) := -1$.
      Since for any \(j\in\range n\), \(y^{(l_{i-1})}_j \in\{1,2\}\) exactly when \(x^{i-1}_j \notin\B\), the new $F$ is a refinement of $f$.
\(\forall 0 < j \leq l_i\), \(y^{(j-1)} \ite F y^{(j)}\) holds in the new $F$
      as the $e$-th component never increases value beyond $2$ along \(y,y',\dots,y^{(l_{i})}\).

      Finally, \(y,y',\dots,y^{(l_{i})}\) is indeed a trace refinement of
      \(x,x',\dots, x^{(i)}\) with $\kappa$:
      \begin{enumerate}
        \item $\kappa$ being non-decreasing is guaranteed as $l_i > l_{i-1}$.
        \item $\kappa(0) = 0$ remains unchanged from the initial step
          and $\kappa(i) = l_i$ by definition.
        \item $0 < y^{(l_i)}_e = 1 < 3$
          and the rest follows from the induction hypothesis.
        \item ${F}_{e}(y^{(\kappa(i - 1))}) = {F}_{e}(y^{(l_{i-1})}) = -1$.
      \end{enumerate}
    \item ${\pi}_{i}(k) \in \B$.
      No change is made safe for the completion of the trace refinement function
      $\kappa: i \mapsto \kappa(i - 1)$.

      \(y,y',\dots,y^{(l_{i-1})}\) being a trace refinement of \(x,x',\dots, x^{(i)}\)
      is trivial as the fourth point of Definition~\ref{def:trace_refinement} is not applicable.
  \end{itemize}
\end{proof}

\subsection{Computational complexity}
We address the computational complexity of basic dynamical properties with the most permissive
semantics.
Complexity with usual (a)synchronous semantics is given in Appendix~\ref{appendix:complexity}

\begin{definition}[Fixed point]
    A configuration \(x\in\B^n\) is a \emph{fixed point} of the BN \(f\) with semantics \(\some\)
    whenever \[\dReach\some f(x)=\{x\}\enspace.\]
\end{definition}

\begin{definition}[Reachability]
    Given two configurations \(x,y\in\B^n\) of a BN \(f\) with semantics \(\some\),
    \(y\) is \emph{reachable} from \(x\) whenever \[y\in\dReach\some f(x)
    \enspace.\]
\end{definition}

\begin{definition}[Attractor]
    A non-empty set of configurations \(A\subseteq\B^n\) is an \emph{attractor}  of the BN \(f\)
    with semantics \(\some\) whenever
    \[\forall x,y\in A, \quad \dReach\some f(x)=\dReach\some f(y)
    \enspace.\]
\end{definition}

First, remark that fixed points of the most permissive semantics are exactly the fixed points of
\(f\): for any configuration \(x\in\mpDom^n\), \(\mpReach f(x) = \{x\} \Leftrightarrow x\in\B^n\wedge f(x)=x\).
Therefore the complexity of deciding if a configuration \(x\) is a fixed point is NP-complete
(Proposition~\ref{pro:fixpoint-complexity}).

\subsubsection{Reachability}
Lemma~\ref{lem:reachability} establishes that if there exists a sequence of most-permissive
transitions from a configuration \(x\) to a configuration \(y\), then there exists a sequence of
linear length linking the two configurations.
Lemma~\ref{lem:reach-try} then states that searching for such a sequence requires exploring at
most a quadratic number of transitions, which leads to Theorem~\ref{thm:complexity-reachability}
establishing the computational complexity for deciding reachability as in P for locally-monotonic
BNs and in P\(^\text{NP}\) (also known as \(\Delta^P_2\)) otherwise.

\def\stageA{\hat z}
\def\stageB{\check z}
\begin{lemma}\label{lem:reachability}
Given a BN $f$ of dimension $n$ and any configurations $x,y\in\B^n$,
    if $x\mpreach f y$, then there exists a sequence of at most $3n$ transitions $\mpite f$ from $x$ to $y$.
    This sequence starts with at most $n$ and at least $\card{\Delta(x,y)}$ transitions of the form
$\B \to \{\mpup,\mpdw\}$,
then at most $n$ transitions of the form
$\{\mpup,\mpdw\} \to \{\mpdw,\mpup\}$,
and then at most $n$ transitions of the form
$\{\mpup,\mpdw\} \to \B$.
\end{lemma}
\begin{proof}
Let us consider any sequence of transitions
$x \mpite f w^1 \mpite f \cdots w^k \mpite f y$.
Let us define the set of components which went through the state \(\mpup\) or \(\mpdw\) during this sequence of
transitions,
$\hat I\DEF\{ i\in\range n\mid \exists j\in\range k, w^j \notin\B \}.$

Let us prove that there exists $\stageA\in\mpDom^n$ with $\Delta(x,\stageA)=\hat I$ and $\forall
i\in\hat I$, $\stageA_i\notin\B$, such that $x\mpreach f \stageA$ in $\card{\hat I}$ transitions.
For each component $i\in\hat I$, we write
$\nu(i)$ the smallest index $j\in\range k$ such that \(w^j_i\neq \B\).
Necessarily, for each $i\in\hat I$,
$\exists z\in\mptob(w^{\nu(i)-1}): f_i(z)\neq x_i$,
identifying $w^0$ with $x$.
The components in \(\hat I\) can then be ordered as
$\{i^1, \dots, i^{\card{\hat I}}\} = \hat I$
with
$\nu(i^1) < \dots < \nu(i^{\card{\hat I}})$.
First, remark that $\nu(i^1)=1$, hence $x\mpite f z^1$ with
$\Delta(x,z^1)=\Delta(w^{\nu(i^1)-1},w^{\nu(i^1)}) = \{i^1\}$
and $z^1_{i^1}=w^{\nu(i^1)}_{i^1}$.
Then, remark that $\mptob(w^{\nu(i^2)})\subseteq\mptob(z^1)$,
hence,
$z^1\mpite f z^2$ with
$\Delta(z^1,z^2)=\Delta(w^{\nu(i^2)-1},w^{\nu(i^2)}) = \{i^2\}$
and $z^2_{i^2}=w^{\nu(i^2)}_{i^2}$.
By induction,
we obtain $x\mpreach f \hat z$.
Remark that \(\forall i\in\hat I\),
\(\hat z_i=\mpup\) whenever \(x_i=0\)
and \(\hat z_i=\mpdw\) whenever \(x_i=1\).

Now, let us consider the subset of components in \(\hat I\) which are equal in $x$ and $y$,
$\bar I\DEF\{ i\in\hat I\mid x_i=y_i \}$:
for each of these components \(i\in\bar I\), there exists \(j'\in \nrange{\nu(i)}k\) such that
\(w^{j'}_i=\mpdw\) whenever \(x_i=y_i=0\) and
\(w^{j'}_i=\mpup\) whenever \(x_i=y_i=1\).
By definition of \(\hat I\) and \(\hat z\), we obtain that
\(\mptob(w^{j'})\subseteq \mptob(\hat z)\).
Therefore, there exists $\stageB\in\mpDom^n$ with $\Delta(\stageA,\stageB)=\bar I$ and
\(\stageA \mpreach f\stageB\) using \(\card{\bar I}\) transitions.
Finally, remark that $\stageB\mpreach f y$ using $\card{\hat I}$ transitions.

\noindent
In summary, $x\mpreach f \stageA\mpreach f\stageB\mpreach f y$ in $\card{\hat I}+\card{\bar
I}+\card{\hat I}\leq 3n$ iterations.
\end{proof}

\begin{lemma}\label{lem:reach-try}
    Given a BN $f$ of dimension $n$ and any configurations
    $x,y\in\B^n$, deciding if $x\mpreach f y$ requires computing
    at most $\frac{n(n-1)}2$ transitions of $\mpite f$;
    whenever \(y\) belongs to an attractor, it requires as most \(n\) transitions.
\end{lemma}
\begin{proof}
    \def\stageAL{\stageA^{L}}
    Let us consider the following procedure with $L\subseteq\range n$, initially
    with $L=\emptyset$:
    \begin{enumerate}
    \item \label{it:build}
        From $x$, apply only transitions of the form $\B\to \{\mpup,\mpdw\}$ to components
            $i\in\range n\setminus L$.
    Let us denote by $\stageAL\in\mpDom^n$ the (unique) reached configuration.
    \item \label{it:reject} If $y\notin\mptob(\stageAL)$, then $y$ is not reachable from $x$.
    \item Otherwise, let us consider the components that cannot reach their value in $y$ from
        $\stageAL$,
    $\bar I^L\DEF\{ i\in\range n\mid \stageAL_i\notin \B\wedge \nexists z\in\mptob(\stageAL),
    f_i(z)=y_i \}$:
    \begin{enumerate}
        \item \label{it:accept} If $\bar I^L=\emptyset$, then $\stageAL\mpreach f y$.
        \item \label{it:iterate} Otherwise, repeat the procedure with $L:=L\cup\bar I^L$.
    \end{enumerate}
    \end{enumerate}
    Remark that this procedure can be iterated at most $n$ times, each of them computing at most
    $n-\card L$ transitions.
    Its correctness can be demonstrated as follows.

    By Lemma~\ref{lem:reachability}, $x\mpreach f y$ if and only if there exists $L\subseteq\range n$ such that
    $y\in\mptob(\stageAL)$ and $\bar I^L=\emptyset$.
    Notice that there is a unique $\subseteq$-minimal $L^*$ verifying
    $y\in\mptob(\stageA^{L^*})$ and
    $\bar I^{L^*}=\emptyset$:
    if $L^1$ and $L^2$ verify these properties, then so does $L^1\cap L^2$.

    Let us denote by $L^0,\dots,L^m$ the successive values of $L$ at the beginning of each iteration
    of the procedure ($L^0=\emptyset$).
    We prove that $L^*=L^m$.
    Let us admit that $L^k\subseteq L^*$ with $k<m$.
    By construction, $\mptob(\stageA^{L^*})\subseteq\mptob(\stageA^{L^k})$.
    Let us assume there exists $i\in\bar I^{L^k}$ and $i\notin L^*$.
    Then, $\stageA^{L^*}=\stageA^{L^k}\notin\B$, and there exists
    $z\in\mptob(\stageA^{L^*})$ with $f_i(z)=y_i$, which is a contradiction.

    Whenever \(y\) belongs to an attractor, \(\bar I^\emptyset=\emptyset\).
    Indeed, remark that \(\mpReach f(y)\subseteq\mptob(\stageA^\emptyset)\).
    Thus, if there exists a component \(i\in\bar I^\emptyset\), then
    from any configuration \(y'\in\mpReach f(y)\), \(y\notin\mpReach f(y')\), which is a
    contradiction.
    Therefore, the procedure is executed only once, which involves computing at most \(n\)
    transitions.
\end{proof}

Steps 1 and 3 of the procedure check for the existence of transitions in a most-permissive
configuration, i.e., for the existence of a binary configuration compatible with it and such that
the local function has a given value.
This is exactly the SAT problem, which is NP-complete in the general case, and P whenever \(f\) is
locally-monotonic.

\begin{theorem}\label{thm:complexity-reachability}
    Given a BN \(f\) of dimension \(n\) and two configurations \(x,y\in\B^n\), deciding if
    \(y\in\mpReach f(x)\) is in P if \(f\) is locally-monotonic,
    and in P\(^{\text{NP}}\) otherwise.
\end{theorem}

\subsubsection{Attractors}
Attractors of the BN \(f\) with the most permissive semantics match exactly with the
\textit{minimal trap spaces} of \(f\)~\cite{Klarner15-TrapSpaces}
(Proposition~\ref{pro:attractors}).
Thus, Deciding if a given configuration \(x\) belongs to an boils down to deciding if the smallest
hypercube closed by \(f\) and containing \(x\) is minimal.

The fact that an attractor is necessarily an hypercube comes from the property that if two
configurations lying on a diagonal of an hypercube are within the same attractor, then all adjacent
configurations are within the attractor as well.
This is illustrated by the following drawing, where boxed configurations belongs to a same
attractor:
\begin{center}
    \begin{tabular}{m{4cm}cm{4cm}}
        \begin{tikzpicture}
\matrix[column sep=0.8cm, row sep=1cm,gray] {
    \node[black] (s010) {$\bf 010$}; \&
      \node[black,draw] (s110) {$\bf 110$};
\\
\node[black,draw] (s000) {$\mathbf{000}$}; \&
    \node[black] (s100) {$\bf 100$};
\\
};
\matrix[column sep=0.8cm, row sep=1cm,shift={(1cm,0.6cm)},gray] {
  \node (s011) {$011$}; \&
  \node (s111) {$111$};
\\
  \node (s001) {$001$}; \&
  \node (s101) {$101$};
\\
};
\path[gray]
    (s000) edge (s010) edge (s100) edge[densely dashed] (s001)
    (s110) edge (s010) edge (s100) edge (s111)
    (s001) edge[densely dashed] (s011) edge[densely dashed] (s101)
    (s111) edge (s011) edge (s101)
    (s010) edge (s011)
    (s100) edge (s101)
    ;
        \end{tikzpicture}&
        \(\Longrightarrow\)
        &
\begin{tikzpicture}
\matrix[column sep=0.8cm, row sep=1cm,gray] {
    \node[black,draw] (s010) {$\bf 010$}; \&
      \node[black,draw] (s110) {$\bf 110$};
\\
\node[black,draw] (s000) {$\mathbf{000}$}; \&
    \node[black,draw] (s100) {$\bf 100$};
\\
};
\matrix[column sep=0.8cm, row sep=1cm,shift={(1cm,0.6cm)},gray] {
  \node (s011) {$011$}; \&
  \node (s111) {$111$};
\\
  \node (s001) {$001$}; \&
  \node (s101) {$101$};
\\
};
\path[gray]
    (s000) edge (s010) edge (s100) edge[densely dashed] (s001)
    (s110) edge (s010) edge (s100) edge (s111)
    (s001) edge[densely dashed] (s011) edge[densely dashed] (s101)
    (s111) edge (s011) edge (s101)
    (s010) edge (s011)
    (s100) edge (s101)
    ;
        \end{tikzpicture}   \end{tabular}
    \end{center}

\begin{proposition}\label{pro:attractors}
    \(A\subseteq \B^n\) is an attractor of  \(f\) with the most permissive semantics if and only if 
    there exists a minimal hypercube \(h\in\H^n\) closed by \(f\) such that \(c(h)=A\).
\end{proposition}
\begin{proof}
Let us consider a configuration \(x\in A\), and
let \(h\in\H^n\) be the smallest hypercube closed by \(f\) containing \(x\).
Let us denote by \(y\in c(h)\) the configuration of this hypercube which is the most distant from \(x\):
$\forall i\in\range n$, \(y_i=\neg x_i\) whenever \(h_i=* \), otherwise \(y_i=x_i\).
According to the previous section, remark that \(y\) is reachable from \(x\); thus by attractor
hypothesis, \(x\) is also reachable from \(y\).
Now remark that the smallest hypercube closed by \(f\) and containing \(y\) is the same \(h\)
(otherwise \(h\) would not be closed).
Thus, for any component \(i\in\range n\) which is free in \(h\) (\(h_i=*\)), there exists a
configuration \(z\in c(h)\) such that \(f_i(z)=0\) and a configuration \(z'\in c(h)\) such that \(f_i(z')=1\).
Therefore, any configuration of the hypercube \(h\) is reachable from \(x\).
Finally, notice that if there exists a configuration \(y\in c(h)\) reachable from \(x\), but where
\(x\) does not belong to the smallest hypercube closed by \(f\) and containing \(y\),
then \(x\) is not reachable from \(y\), and thus does not belong to any attractor.
\end{proof}

\begin{theorem}
    Given a BN \(f\) of dimension \(n\) and a configuration \(x\in\B^n\),
    deciding if \(x\) belongs to an attractor of \(f\) with the most permissive semantics is in
coNP whenever \(f\) is locally monotonic, and in coNP\(^{\text{coNP}}\) otherwise.
\end{theorem}
\begin{proof}
    Consider IS-NOT-CLOSED(\(f,h\)) the problem of deciding if the given hypercube \(h\) is
    \emph{not} closed by \(f\):
    it is equivalent to deciding if there exists component
\(i\in\range n\) with \(h_i\neq *\) and \(z\in c(h)\) such that \(f_i(z)\neq h_i\), which is
NP-complete in general, and P whenever \(f\) is locally monotonic.
Then, the complementary problem IS-CLOSED(\(f,h\)) is in coNP in the general case and in P in the
locally-monotonic case.

    Consider IS-NOT-MINIMAL(\(f,h\)) the problem of deciding if the hypercube \(h\) closed by \(f\)
    is \emph{not} minimal: it can be solved by deciding wherever there exists an hypercube \(h'\) which is
    strictly included in \(h\) and which is closed by \(f\), which is at most NP$^{\text{IS-CLOSED}}$.
    Thus, the complementary problem IS-MINIMAL(\(f,h\)) is in coNP$^{\text{IS-CLOSED}}$, i.e.,
coNP$^{\text{coNP}}=\Pi^{\text P}_2$ in the general case and coNP in the locally-monotonic case.
\end{proof}

\section{Discussion}

The characterization of reachability and attractors with the most permissive
semantics matches with prior introduced approximations for BNs:
the reachability analysis in most permissive semantics is very close to the
\emph{meta-state} semantics of \cite{Caspots-BioSystems16} which was introduced as an
over-approximation of reachability in BNs with (generalized) asynchronous update.
Moreover, the attractors of the most permissive semantics match with the
\emph{minimal trap spaces} \cite{Klarner15-TrapSpaces} of BNs, which are used to
approximate attractors in BNs with asynchronous update (which can be different from hypercubes).

Dynamics of BNs with usual updating modes is often represented with state transition graphs, where
nodes are the Boolean configurations (states), and edges represent the possible iterations (transitions).
Such an object is much less relevant with the most permissive semantics as there would be a direct
transition from a configuration to each of the configurations reachable from it.
Alternatively, hierarchies of trap spaces (hypercubes), as described in \cite{Klarner15-TrapSpaces}
constitutes a more promising structure to visualize the attractor basins and undergoing differentiation
processes.

The model refinement criteria we consider is very general and aims at introducing as little as
biases as possible without extra information.
Nevertheless, exploring different sub-classes of admissible model refinements of BNs and define minimal
Boolean semantics capturing them constitutes a challenging research direction.

An extended discussion and assessment of the most permissive semantics for the modeling of
biological networks can be found is \cite{MPBNs}.
Case studies are given in Appendix~\ref{appendix:case-studies}.

\bibliographystyle{unsrturl}
\bibliography{techreport}

\begin{thebibliography}{10}

\bibitem{MPBNs}
Lo{\"\i}c Paulev{\'e}, Juraj Kol{\v c}{\'a}k, Thomas Chatain, and Stefan Haar.
\newblock Reconciling qualitative, abstract, and scalable modeling of
  biological networks.
\newblock {\em bioRxiv}, 2020.
\newblock \href {http://dx.doi.org/10.1101/2020.03.22.998377}
  {\path{doi:10.1101/2020.03.22.998377}}.

\bibitem{Klarner15-TrapSpaces}
Hannes Klarner, Alexander Bockmayr, and Heike Siebert.
\newblock Computing maximal and minimal trap spaces of {B}oolean networks.
\newblock {\em Natural Computing}, 14(4):535--544, 2015.
\newblock \href {http://dx.doi.org/10.1007/s11047-015-9520-7}
  {\path{doi:10.1007/s11047-015-9520-7}}.

\bibitem{Caspots-BioSystems16}
Max Ostrowski, Lo{\"i}c Paulev{\'e}, Torsten Schaub, Anne Siegel, and Carito
  Guziolowski.
\newblock Boolean network identification from perturbation time series data
  combining dynamics abstraction and logic programming.
\newblock {\em Biosystems}, 149:139 -- 153, 2016.
\newblock \href {http://dx.doi.org/10.1016/j.biosystems.2016.07.009}
  {\path{doi:10.1016/j.biosystems.2016.07.009}}.

\bibitem{Formenti2014}
Enrico Formenti, Luca Manzoni, and Antonio~E. Porreca.
\newblock Fixed points and attractors of reaction systems.
\newblock In {\em Language, Life, Limits}, pages 194--203. Springer
  International Publishing, 2014.
\newblock \href {http://dx.doi.org/10.1007/978-3-319-08019-2_20}
  {\path{doi:10.1007/978-3-319-08019-2_20}}.

\bibitem{Dennunzio2019}
Alberto Dennunzio, Enrico Formenti, Luca Manzoni, and Antonio~E. Porreca.
\newblock Complexity of the dynamics of reaction systems.
\newblock {\em Information and Computation}, 267:96--109, aug 2019.
\newblock \href {http://dx.doi.org/10.1016/j.ic.2019.03.006}
  {\path{doi:10.1016/j.ic.2019.03.006}}.

\bibitem{Nakamura1981}
Katsuhiko Nakamura.
\newblock Synchronous to asynchronous transformation of polyautomata.
\newblock {\em Journal of Computer and System Sciences}, 23(1):22--37, August
  1981.
\newblock \href {http://dx.doi.org/10.1016/0022-0000(81)90003-9}
  {\path{doi:10.1016/0022-0000(81)90003-9}}.

\bibitem{gekakasc12a}
M.~Gebser, R.~Kaminski, B.~Kaufmann, and T.~Schaub.
\newblock {\em Answer Set Solving in Practice}.
\newblock Synthesis Lectures on Artificial Intelligence and Machine Learning.
  Morgan and Claypool Publishers, 2012.

\bibitem{clingo}
Martin Gebser, Roland Kaminski, Benjamin Kaufmann, and Torsten Schaub.
\newblock Clingo = {ASP} + control: Preliminary report.
\newblock {\em CoRR}, abs/1405.3694, 2014.

\bibitem{ColomotoNotebook2018}
Aur{\'e}lien Naldi, C{\'e}line Hernandez, Nicolas Levy, Gautier Stoll, Pedro~T.
  Monteiro, Claudine Chaouiya, Tom{\'a}{\v s} Helikar, Andrei Zinovyev,
  Laurence Calzone, Sarah Cohen-Boulakia, Denis Thieffry, and Lo{\"i}c
  Paulev{\'e}.
\newblock {The CoLoMoTo Interactive Notebook: Accessible and Reproducible
  Computational Analyses for Qualitative Biological Networks}.
\newblock {\em {Frontiers in Physiology}}, 9:680, 2018.
\newblock \href {http://dx.doi.org/10.3389/fphys.2018.00680}
  {\path{doi:10.3389/fphys.2018.00680}}.

\bibitem{Cohen2015}
David P.~A. Cohen, Loredana Martignetti, Sylvie Robine, Emmanuel Barillot,
  Andrei Zinovyev, and Laurence Calzone.
\newblock Mathematical modelling of molecular pathways enabling tumour cell
  invasion and migration.
\newblock {\em PLoS Comput Biol}, 11(11):e1004571, 2015.
\newblock \href {http://dx.doi.org/10.1371/journal.pcbi.1004571}
  {\path{doi:10.1371/journal.pcbi.1004571}}.

\bibitem{Abou-Jaoude2015}
Wassim Abou-Jaoud{\'e}, Pedro~T. Monteiro, Aur{\'e}lien Naldi, Maximilien
  Grandclaudon, Vassili Soumelis, Claudine Chaouiya, and Denis Thieffry.
\newblock Model checking to assess {T}-helper cell plasticity.
\newblock {\em Frontiers in Bioengineering and Biotechnology}, 2, 2015.
\newblock \href {http://dx.doi.org/10.3389/fbioe.2014.00086}
  {\path{doi:10.3389/fbioe.2014.00086}}.

\bibitem{Didier11}
Gilles Didier, Elisabeth Remy, and Claudine Chaouiya.
\newblock Mapping multivalued onto {B}oolean dynamics.
\newblock {\em Journal of Theoretical Biology}, 270(1):177 -- 184, 2011.
\newblock \href {http://dx.doi.org/10.1016/j.jtbi.2010.09.017}
  {\path{doi:10.1016/j.jtbi.2010.09.017}}.

\bibitem{Remy2015}
Elisabeth Remy, Sandra Rebouissou, Claudine Chaouiya, Andrei Zinovyev, Fran{\c
  c}ois Radvanyi, and Laurence Calzone.
\newblock A modeling approach to explain mutually exclusive and co-occurring
  genetic alterations in bladder tumorigenesis.
\newblock {\em Cancer Research}, 75(19):4042--4052, aug 2015.
\newblock \href {http://dx.doi.org/10.1158/0008-5472.can-15-0602}
  {\path{doi:10.1158/0008-5472.can-15-0602}}.

\bibitem{Mendes2018}
Nuno~D. Mendes, Rui Henriques, Elisabeth Remy, Jorge Carneiro, Pedro~T.
  Monteiro, and Claudine Chaouiya.
\newblock Estimating attractor reachability in asynchronous logical models.
\newblock {\em Frontiers in Physiology}, 9, 2018.
\newblock \href {http://dx.doi.org/10.3389/fphys.2018.01161}
  {\path{doi:10.3389/fphys.2018.01161}}.

\bibitem{VLBNs-v1}
Loïc Paulevé.
\newblock {VLBNs - Very Large Boolean Networks (Version 1) [dataset]}, 2020.
\newblock Zenodo. \url{https://doi.org/10.5281/zenodo.3714876}.

\bibitem{Albert2002}
R{\'{e}}ka Albert and Albert-L{\'{a}}szl{\'{o}} Barab{\'{a}}si.
\newblock Statistical mechanics of complex networks.
\newblock {\em Reviews of Modern Physics}, 74(1):47--97, 2002.
\newblock \href {http://dx.doi.org/10.1103/revmodphys.74.47}
  {\path{doi:10.1103/revmodphys.74.47}}.

\end{thebibliography}

\clearpage
\appendix

\section{Complexity of Dynamical Properties with (A)synchronous Semantics}
\label{appendix:complexity}

Let us fix a BN \(f\) of dimension \(n\).

\subsection{Update semantics}

BN semantics are expressed as irreflexive binary relations between the configurations.
We use the symbol \(\rightarrow\) decorated with the Boolean function and a symbol representing the
semantics.

\begin{definition}[Synchronous semantics]
    \begin{equation*}
        \forall x,y\in\B^n\quad
        x\site f y\EQDEF x\neq y\wedge y=f(x)
        \enspace.
    \end{equation*}
\end{definition}

\begin{definition}[Fully asynchronous semantics]
\begin{equation*}
\forall x,y\in\B^n,\quad
x\fite f y \EQDEF \exists i\in\range n: \Delta(x,y)=\{i\}\wedge y_i=f_i(x)
    \enspace.
\end{equation*}
\end{definition}

\begin{definition}[Asynchronous semantics]
\begin{equation*}
\forall x,y\in\B^n,\quad
    x\ite f y \EQDEF x\neq y\wedge \forall i\in\Delta(x,y), y_i=f_i(x)
    \enspace.
\end{equation*}
\end{definition}

Given a semantics \(\some\), we write
 \(\dite\some f\closure\) the reflexive and transitive closure of the binary relation \(\dite\some f\).
 Thus, \(x\dite\some f\closure y\) if and only if \(x=y\) or there exists a sequence
 \(x\dite\some f x' \dite\some f \cdots\dite\some f y\).
 The set of configurations which are in such a relation with a configuration \(x\) is given by
 \(\dReach\some f(x)\):
\begin{equation}
    \dReach\some f(x) \DEF \{ y \in\B^n\mid x\dite\some f\closure y\}\enspace.
\end{equation}

\subsection{Fixed points}
Remark that with synchronous, fully asynchronous, and asynchronous semantics, \(x\in\B^n\) is a
fixed point if and only if \(f(x)=x\).

\begin{proposition}\label{pro:fixpoint-complexity}
    Deciding if there exists \(x\in\B^n\) such that \(f(x)=x\) is NP-complete.
\end{proposition}
\begin{proof}
By reduction of the SAT problem \cite{Formenti2014}.
\end{proof}

\subsection{Reachability}

\begin{proposition}
    Given two configurations \(x,y\in\B^n\),
    deciding if \(y\in\dReach\some f(x)\) with \(\some\in\{\us,\uf,\ua\}\) is PSPACE-complete.
\end{proposition}
\begin{proof}
As there is at most \(2^n\) configurations to explore, the problem is at most in PSPACE, as it
sufficient to apply non-deterministically at most \(2^n-1\) transitions from \(x\) using a
counter on \(n\) bits.

With the synchronous semantics, the PSPACE-hardness derives by reduction of the
reachability problem in reaction systems, a subclass of synchronous BNs \cite{Dennunzio2019}.

With fully asynchronous and asynchronous semantics, the PSPACE-hardness derives by reduction of the
reachability problem in synchronous BNs.
Indeed, similarly to cellular automata \citep{Nakamura1981},
one can define a BN \(f'\) so that asynchronous and fully asynchronous semantics give reachability
relations that are equivalent with the synchronous semantics of \(f\).

A possible construction is to decompose a synchronous transition in several steps which can be
performed asynchronously.
This involves 3 stages:
(a) the computation of the next value for each component \(i\in\range n\);
(b) the application of the new state for each component;
(c) the reset of components introduced by the construction.
We give here an encoding as a BN \(f'\) with \(3n+2\) dimensions:
    one component \(\mathrm z\) for which the state \(1\) triggers the reset of additional
    components (except \(\mathrm z\));
    one component \(\mathrm w\) for which the state \(0\), assuming \(\mathrm z\) has state
    \(0\), triggers the computation stage (a), and the state \(1\) triggers the application stage
    (b).
    For each component \(i\in\range n\) of \(f\), two components \(\mathrm ci\)
    and \(\bar{\mathrm c}i\) are defined, for which the state \(1\) specify respectively if \(f_i\)
    is true or false.
    The end of computation stage (a) is detected whenever for each component \(i\in\range n\),
    either \(\mathrm ci\) or \(\bar{\mathrm c}i\) are in state \(1\).
    Component \(\mathrm w\) then switch to state \(1\);
    then, components \(i\) switch to state \(0\) if and only if \(\bar{\mathrm c}i\) is
    \(1\) and to state \(1\) if and only if \(\mathrm ci\) is \(1\).
    The end of application stage (b) is detected whenever all the components \(i\in\range n\) have
    been updated.
    Component \(\mathrm z\) then witch to state \(1\) which will trigger the switch to state 
    \(0\) of components \(\mathrm w\), \(\mathrm ci\) and \(\bar{\mathrm c}i\).
    Finally, component \(z\) switch back to state \(0\), which allows components \(\mathrm
    ci\) and \(\bar{\mathrm c}i\) computing the next state of each components \(i\in\range
    n\).
    This network \(f':\B^{3n+2}\to\B^{3n+2}\) can be formally defined as follows, where \(x_{1..n}\)
    denotes the configuration \(x\) truncated at the first \(n\) components:
    \begin{align*}
        f'_i(x') &= ((\neg x'_{\mathrm w} \vee x'_{\mathrm z}) \wedge x'_i) \vee (x'_{\mathrm w} \wedge \neg x'_{\mathrm z} \wedge x'_{\mathrm ci})
        \\
        f'_{\mathrm ci}(x') &=
            \neg x'_{\mathrm z} \wedge 
            \left(
            (\neg x'_{\mathrm w} \wedge f_i(x'_{1..n}))
            \vee
            (x'_{\mathrm w}\wedge x'_{\mathrm ci})\right)
        \\
        f'_{\bar{\mathrm c}i}(x') &=
            \neg x'_{\mathrm z} \wedge 
            \left(
            (\neg x'_{\mathrm w} \wedge \neg f_i(x'_{1..n}))
            \vee
            (x'_{\mathrm w}\wedge x'_{\bar{\mathrm c}i})\right)
        \\
        f'_{\mathrm w}(x') &=
        \neg x'_{\mathrm z} \wedge \left( \left(x'_{\mathrm w} \vee \textstyle
                \bigwedge_{i\in\range n}
                (x'_{\mathrm ci}\vee x'_{\bar{\mathrm c}i})\right)\right)
        \\
        f'_{\mathrm z}(x') &= \textstyle
        \left(x'_{\mathrm w} \wedge 
                \bigwedge_{i\in\range n}
        (x'_{\mathrm ci}\Leftrightarrow x'_{i} \wedge
        x'_{\bar{\mathrm c}i}\Leftrightarrow \neg x'_{i})\right) 
            \vee
            \left(x'_{\mathrm z} \wedge \left(x'_{\mathrm w}\vee \bigvee_{i\in\range n}
                (x'_{\mathrm ci}\vee x'_{\bar{\mathrm c}i})\right)\right)
    \end{align*}
    It results that for all pairs of configurations \(x,y\in\B^n\),
    \begin{equation*}
    y\in\dReach\us f(x)
    \Longleftrightarrow y0^{2n+2} \in\dReach\ua {f'}(x0^{2n+2})
\Longleftrightarrow y0^{2n+2} \in\dReach\uf {f'}(x0^{2n+2})
\end{equation*}
    where
    \(0^{2n+2}\) is the \(0\) vector of dimension \(2n+2\)
    and \(y0^{2n+2}\) denotes its concatenation to \(y\).
\end{proof}

\subsection{Attractors}

\begin{proposition}
    Given a configuration \(x\in\B^n\),
    deciding if \(x\) belongs to an attractor of \(f\) with semantics
    \(\some\in\{\us,\uf,\ua\}\) is PSPACE-complete.
\end{proposition}
\begin{proof}
    The problem is in PSPACE as one can solve it through its complementary:
    $x$ does not belong to any attractor if and only if there exists a configuration $y$ such that
    \(y\) is reachable from \(x\) and \(x\) is not reachable from \(y\).

    For the synchronous semantics, the PSPACE-hardness derives by reduction of the same problem in synchronous
    reaction systems, a subclass of BNs \cite{Dennunzio2019}.
    Then, the proof can be lifted to asynchronous and fully asynchronous semantics by the reduction
    of the problem with the synchronous semantics (e.g., by using the construction in the previous
    section).
\end{proof}

\clearpage
\section{Case studies}
\label{appendix:case-studies}

\def\tag{2020-03-19}

\subsection{Code}
A simple implementation of reachability and attractor computations for Most Permissive Boolean
Networks is available at \url{https://github.com/pauleve/mpbn}.
It relies on Answer-Set Reprogramming \cite{gekakasc12a} and the solver \texttt{clingo}
\cite{clingo} which offers features such as minimal model enumeration.

The computational analyzes have been performed within the CoLoMoTo environment
\cite{ColomotoNotebook2018} and can thus be reproduced using the provided notebook files within the
Docker image \texttt{colomoto/colomoto-docker:\tag}:

Using Python (\url{https://python.org}), execute the following command in a terminal:
\begin{lstlisting}
sudo pip install -U colomoto-docker # you may have to use pip3
colomoto-docker -V (*\tag*)
\end{lstlisting}
Alternatively, you can run the image directly with Docker (\url{https://docker.com}):
\begin{lstlisting}
docker run -it --rm -p 8888:8888 colomoto/colomoto-docker:(*\tag*)
\end{lstlisting}
and then open your webbrowser to \url{https://localhost:8888}.
See \url{https://colomoto.org/notebook} for detailed instructions.

The notebook files (with \texttt{.ipynb} extension) used in the next sections can be downloaded from
\url{http://doi.org/10.5281/zenodo.3719097} and then be uploaded and executed within the Jupyter web interface.

\subsection{Models of differentiation processes from literature}
We show on several case studies from literature that MPBNs, although potentially predicting more
behaviors than asyncrhonous BNs, are still stringent enough to predict cell fate decision processes,
i.e., absence of attractor reachable from specific configurations or with specific perturbations.

\subsubsection{Tumour invasion model by Cohen et al. 2015}
This BN of 32 components \cite{Cohen2015} models cellular decision processes involved in tumour invasion, with
attractors related to apoptosis, cell cycle arrest, and various stages leading to metastasis.
We reproduced%
\footnote{Notebook ``\textit{MPBN applied to Tumour invasion model by Cohen et al. 2015.ipynb}'',
visualize online at
\url{https://nbviewer.jupyter.org/gist/pauleve/155737b8efcbc909bca18aadf3520cc0}}
the analysis of reachable attractor in the wild-type model and with p53 and Notch
mutations, and whose combination lead to a loss of capability to reach apoptotic attractors.
Analysis with MPBNs reports the exact same reachable attractors, thus predicting the same as with
fully asynchronous BNs.

\subsubsection{T-cell differentiation model by Abou-Jaoudé et al. 2015}
This multivalued network of 102 components \cite{Abou-Jaoude2015} models reprogramming capabilities across different T-cell types.
We notably reproduced%
\footnote{Notebook ``\textit{MPBN applied to T-Cell differentiation model by Abou-Jaoudé et al.
    2015.ipynb}'', visualize
online at \url{https://nbviewer.jupyter.org/gist/pauleve/88f7e4abadca968d9c7dcdf2c909490c}}
the computation of the \emph{reprogramming graph} between identified cell types and with identified
input conditions, after a booleanization~\citep{Didier11} of the model.
Using MPBNs, the exact same graph is recovered as with fully asynchronous BNs.
The computations have been handled on the original large multivalued model, whereas the original study had to perform approximations through model reduction.
Moreover, the analysis using MPBNs enable devising the nature of all the attractors reachable under the studied conditions
(all fixed-points, except one cyclic under APC condition).

\subsection{Scalability}
The theoretical complexity gain brought by the Most Permissive semantics has a drastic impact for
the analysis of large BNs.
We illustrate it both on large networks from literature, and on very large networks (up to 100,000 components) generated randomly.

\subsubsection{Networks from literature}
In the two analysis of the previous section, the computation of reachable
attractors takes less than 10ms\footnote{Computation times are obtained on an Intel(R) Xeon(R)
E-2124 CPU @ 3.30GHz} on the Tumour invasion model (32 components)
and less than 100ms on the T-cell differentiation model (104 components).

We additionally performed reachable attractor computations on the Bladder tumorigenesis model by Remy et al
2015 \cite{Remy2015}, as it served as benchmark for evaluating simulations methods for devising
reachable attractors and their propensity in \cite{Mendes2018}.
The network has 35 components.
The computation of reachable attractors in diverse settings%
\footnote{Notebook ``\textit{MPBN applied to Bladder Tumorigenesis by Remy et al 2015.ipynb}'', visualize
online at \url{https://nbviewer.jupyter.org/gist/pauleve/b63d362e0df538b022483226a724b408}}
are performed in less than 10ms, whereas simulations of asynchronous BNs were reported taking from
10s to 700s, and possibly not successful.
Note however that contrary to the simulation methods, we are not able to compute nor estimate the
propensity of reachable attractors. Nevertheless, the enumeration is guaranteed to be complete.
In this application, the number of attractors of Most Permissive semantics is the same as the
reported with fully asynchronous semantics.

\subsubsection{Very large networks}
We generated random influence graphs~\cite{VLBNs-v1} with scale-free structure \cite{Albert2002} with a number of
components ranging from 1,000 to 100,000 with in-degrees up to 1,400.
We then applied the inhibitor dominant rule to assign a Boolean function to each component:
the activation occurs only in configurations whenever no inhibitor are active and at least one
activator is active.
The following computations are measured%
\footnote{
    Notebook ``Scalability on large random BNs.ipynb''
    , visualize online at
    \url{https://nbviewer.jupyter.org/gist/pauleve/b7098fb37cac7e318689e41511d10479}.
}:
computation of 1 attractor;
enumeration of (at most) 1000 attractors;
and enumeration of (at most) 1000 attractors reachable from random initial configurations.

Fig.~\ref{fig:scalability-random} summarizes the results and shows that the computation of reachable
attractors take a fraction of a second with 1,0000 components,
less than 2 seconds with 10,000, and less than 50 seconds with 100,000 components.
The computation times exclude the time for parsing the input text file (up to 20s for larger
networks).
\begin{figure}[hp]
\centering
\includegraphics[width=\textwidth]{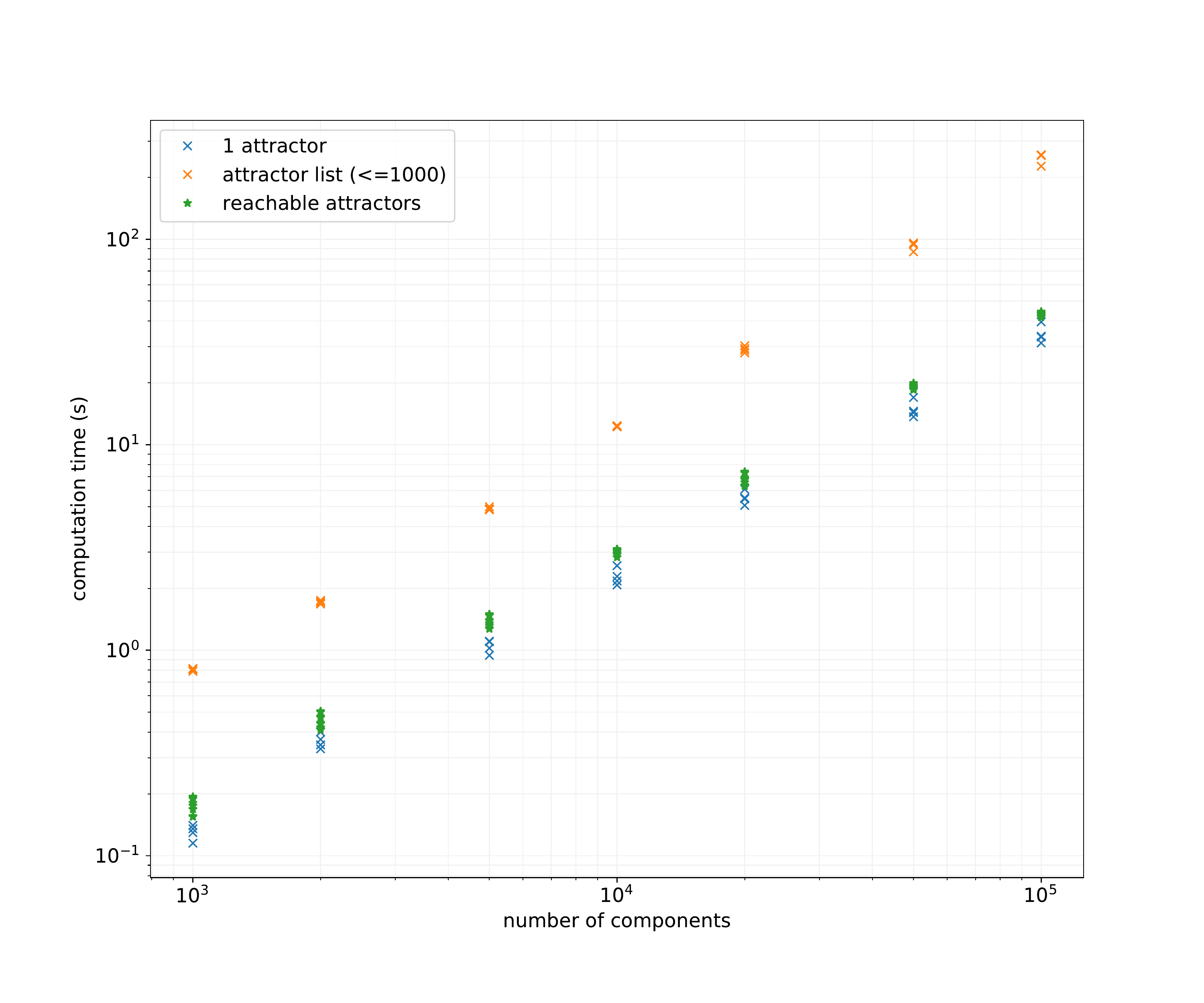}
\caption{Computation times on a 3GHz CPU obtained using \texttt{mpbn} software tool on BNs generated
    with random scale-free influence graph for the computation of a single attractor, the
    enumeration of 1,000 attractor, and the computation of attractors reachable from random initial
conditions.}
\label{fig:scalability-random}
\end{figure}


\end{document}